\documentclass[letterpaper]{article}

\usepackage{geometry}
\usepackage{times}
\usepackage{soul}
\usepackage{url}
\usepackage[hidelinks]{hyperref}
\usepackage[utf8]{inputenc}
\usepackage[small]{caption}
\usepackage{graphicx}
\usepackage{amsmath}
\usepackage{amsthm}
\usepackage{booktabs}
\usepackage{algorithm}
\usepackage{algorithmic}
\newcommand{\citet}[1]{\citeauthor{#1} \shortcite{#1}}
\newcommand{\citep}[1]{\cite{#1}}

\usepackage{wrapfig}
\usepackage{caption}
\usepackage{subcaption}
\usepackage{amsfonts,amssymb, amsthm}
\usepackage{stmaryrd}
\usepackage{thmtools, thm-restate}

\usepackage{natbib}

\usepackage[switch,mathlines]{lineno}

\usepackage{tikz}
\usetikzlibrary{arrows.meta}
\usetikzlibrary{automata, positioning, calc, shapes, arrows, fit}

\theoremstyle{plain}

\newtheorem{proposition}{Proposition}
\newtheorem{definition}{Definition}

\usepackage{fontawesome}

\usepackage{xcolor}
\definecolor{MyOrange}{rgb}{1, 0.4, 0}
\definecolor{MyLightOrange}{rgb}{1, 0.9, 0.7}
\definecolor{MyGray}{rgb}{0.3, 0.3, 0.3}
\definecolor{MyLightGray}{rgb}{0.5, 0.5, 0.5}
\definecolor{MyGreen}{rgb}{0, 0.6, 0}
\definecolor{MyBlue}{rgb}{0, 0, 0.6}
\definecolor{MyLightBlue}{rgb}{0.7, 0.8, 1}
\definecolor{MyRed}{rgb}{0.7, 0, 0}
\definecolor{MyLightRed}{rgb}{1, 0.7, 0.7}
\definecolor{MyLightYellow}{HTML}{f9f6ec}

\usepackage{pifont}
\newcommand{\cmark}{{\color{MyGreen}\ding{51}}}
\newcommand{\xmark}{{\color{MyRed}\ding{55}}}
\newcommand{\cmarkorange}{{\color{MyOrange}\ding{51}}}

\newcommand{\alt}{A}
\newcommand{\prof}[1]{\boldsymbol{#1}} 
\newcommand{\CW}{\textit{CW}}
\newcommand{\majundom}{\textit{MajUD}}
\newcommand{\majdom}{\textit{MajDom}} 
\newcommand{\plurundom}{\textit{PlurUD}}
\newcommand{\plurdom}{\textit{PlurDom}}
\newcommand{\unanundom}{\textit{UnanUD}}
\newcommand{\unandom}{\textit{UnanDom}}
\DeclareMathOperator{\dyn}{MD}
\newcommand{\dyns}[1]{\dyn_{#1}}
\newcommand{\support}[1]{N^{\prof{P}}_{#1}}
\newcommand{\supportz}[2]{N^{\prof{P}^{#1}}_{#2}}

\newcommand\arrowto[1][1.4em]{\tikz[baseline=-0.5ex, 
	shorten <=2pt, shorten >=2pt] \draw[-stealth] (0,0) -- (0.5,0);}
\renewcommand{\emptyset}{\varnothing}
\newcommand{\argmax}{\operatornamewithlimits{argmax}}

\newtheorem{myex}{Example}
\newenvironment{ex}{\begin{myex}\rm}{\hfill$\vartriangle$\end{myex}}
\newtheorem*{myexcon}{Example 1 (continued)} 
\newenvironment{excon}{\begin{myexcon}\rm}{\hfill$\vartriangle$\end{myexcon}}

\let\oldcenter\center
\let\oldendcenter\endcenter
\renewenvironment{center}{\setlength\topsep{0pt}\oldcenter}{\oldendcenter}

\begin{document}
	
	\title{Let's Agree to Agree:
		Targeting Consensus for Incomplete Preferences through Majority Dynamics}
	
	\author{ 
		Sirin Botan$^1$ \hspace{2em} Simon Rey$^2$ \hspace{2em} Zoi Terzopoulou$^3$\footnote{Contact Author}
	}

	\date{
		$^1$UNSW Sydney\\
		$^2$Institute for Logic, Language and Computation (ILLC), University of Amsterdam\\
		$^3$LAMSADE, Université Paris Dauphine-PSL\\[1em]
		s.botan@unsw.edu.au \hspace{2em} s.j.rey@uva.nl. \hspace{2em} zoi.terzopoulou@dauphine.psl.eu
	}
	
	\maketitle
	
	\begin{abstract}
		We study settings in which agents with incomplete preferences need to make a collective decision. We focus on a process of majority dynamics where issues are addressed one at a time and undecided agents follow the opinion of the majority. We assess the effects of this process on various consensus notions---such as the Condorcet winner---and show that in the worst case, myopic adherence to the majority damages existing consensus; yet, simulation experiments indicate that the damage is often mild. We also examine scenarios where the chair of the decision process can control the existence (or the identity) of consensus, by determining the order in which the issues are discussed. 
	\end{abstract}
	
	
	\section{Introduction}
	
	Groups of agents often need to make decisions by finding a consensus between different individual opinions: Amongst friends,
	hiring committees and teams of reviewers, as well as multiagent systems, reaching collective consensus is not always easy. Such situations become even less straightforward when agents hold incomplete opinions,which is often the case in real life, due to uncertainty, lack of knowledge, or reduced interest about the issues in question. Consider an example:
	\begin{ex}
		\label{ex:incomplete-apps}
		Five organisers of an online conference must collectively decide about which video-conferencing app to use:  \emph{AppEar}~($a$), \emph{Bridge}~($b$), or \emph{C--nnect}~($c$)? Having never used Bridge, two organisers hold no opinion on it, but think that AppEar is better than C--nnect (denoted by $a \mathop{\arrowto} c$).  Two organisers familiar with AppEar and Bridge find Bridge superior,  while one organiser has never heard of any app:
		\begin{center}
			\begin{tikzpicture}[scale=0.6]
				\node[anchor = south] (a) at (0,0) {$a$};
				\node[anchor = south] (b) at (0,-1) {$c$};
				\node[anchor = south] (c) at (0.5,0) {$b$};
				
				\draw[-stealth](a)to(b);
				
				\node[anchor = south] (a) at (2,0) {$a$};
				\node[anchor = south] (b) at (2,-1) {$c$};
				\node[anchor = south] (c) at (2.5,0) {$b$};
				
				\draw[-stealth](a)to(b);

				\node[anchor = south] (a) at (4,0) {$b$};
				\node[anchor = south] (b) at (4,-1) {$a$};
				\node[anchor = south] (c) at (4.5,0) {$c$};
				
				\draw[-stealth](a)to(b);
				
				\node[anchor = south] (a) at (6,0) {$b$};
				\node[anchor = south] (b) at (6,-1) {$a$};
				\node[anchor = south] (c) at (6.5,0) {$c$};
				
				\draw[-stealth](a)to(b);
				
				\node[anchor = south] (a) at (8,0) {$a$};
				\node[anchor = south] (b) at (8.5,0) {$b$};
				\node[anchor = south] (c) at (9,0) {$c$};
			\end{tikzpicture}
		\end{center}
		\noindent A consensus alternative for these opinions is not obvious. 
	\end{ex}
	
	\noindent In two rapidly growing streams of literature within computational social choice, \emph{opinion diffusion} and \emph{liquid democracy}, agents 
	adopt the opinions of their peers during collective decision processes---either directly by embracing them or indirectly by delegating their vote \citep{bredereck2017manipulating,brill2018pairwise,botan2019multi,bloembergen2019rational}. In our model, agents also consult their peers, but only about issues on which they initially hold no opinion. Supposing that the group discusses one issue at a time (here the comparison between two alternatives), as commonly happens in practice, all agents with a missing opinion on it will adopt the opinion of the majority. For agents who trust all their peers equally, following majority is indeed the best option, both if they aim at 
	maximising agreement within the group, or at having the highest chances of making the correct decision \citep{condorcet1785essay,may1952set}.\footnote{\citet{rawls1971theory} also advocated the idea that agents hold twofold preferences---in the first level, their intrinsic preferences on the issues; in the second level, they wish to agree with their group.}
	
	Limited by constraints on time and energy, agents do their best to reach consensus \emph{locally} for each single issue that is being discussed, but the group still targets consensus \emph{globally}, after all issues have been addressed. \emph{Majority dynamics} ($\dyn$) can help with that, as illustrated below. 
	
	\begin{excon}
		Suppose that the discussion starts with AppEar in comparison to Bridge. Since only two organisers have a relevant opinion (that $b$ is better than $a$), everyone will adopt it. Those who find $a$ better than $c$ will also rank $b$ above $c$, as a matter of consistency.\footnote{Formally, this consistency requirement is \emph{transitivity}.} Next, suppose that Bridge and C--nnect are discussed. Note that only two organisers rank them, and both rank $b$ above $c$, so this opinion will be adopted by everyone. The last issue addressed is the comparison between AppEar and C--nnect. Again everyone will adopt the only existing opinion: $a$ above $c$.
		
		\begin{center}
			\begin{tikzpicture}
				\node[anchor = south east, label = {[label distance =-0.6em]-90:(\textit{a}) after discussing~$a$ and~$b$}] at (-0.15, 0) {
					\begin{tikzpicture}[scale=0.6, anchor = south]
						\node (a) at (0,0) {$b$};
						\node (b) at (0,-1) {$a$};
						\node (c) at (0,-2) {$c$};
						
						\draw[-stealth](a)to(b);
						\draw[-stealth](b)to(c);
						
						\node (a) at (1,0) {$b$};
						\node (b) at (1,-1) {$a$};
						\node (c) at (1,-2) {$c$};
						
						\draw[-stealth](a)to(b);
						\draw[-stealth](b)to(c);
						
						\node (a) at (2,0) {$b$};
						\node (b) at (2,-1) {$a$};
						\node (c) at (2.5,0) {$c$};
						
						\draw[-stealth](a)to(b);
						
						\node (a) at (3.5,0) {$b$};
						\node (b) at (3.5,-1) {$a$};
						\node (c) at (4,0) {$c$};
						
						\draw[-stealth](a)to(b);
						
						\node (a) at (5,0) {$b$};
						\node (b) at (5,-1) {$a$};
						\node (c) at (5.5,0) {$c$};
						
						\draw[-stealth](a)to(b);
					\end{tikzpicture}
				};
				
				\node[anchor = south west, label = {[label distance =-0.6em]-90:(\textit{b}) after discussing~$b$ and~$c$}] at (0.1, 0) {
					\begin{tikzpicture}[scale=0.6, anchor = south]
						\node (a) at (0,0) {$b$};
						\node (b) at (0,-1) {$a$};
						\node (c) at (0,-2) {$c$};
						
						\draw[-stealth](a)to(b);
						\draw[-stealth](b)to(c);
						
						\node (a) at (1,0) {$b$};
						\node (b) at (1,-1) {$a$};
						\node (c) at (1,-2) {$c$};
						
						\draw[-stealth](a)to(b);
						\draw[-stealth](b)to(c);
						
						\node (a) at (2.5,0) {$b$};
						\node (b) at (2,-1) {$a$};
						\node (c) at (3,-1) {$c$};
						
						\draw[-stealth](a)to(b);
						\draw[-stealth](a)to(c);
						
						\node (a) at (4,0) {$b$};
						\node (b) at (3.5,-1) {$a$};
						\node (c) at (4.5,-1) {$c$};
						
						\draw[-stealth](a)to(b);
						\draw[-stealth](a)to(c);
						
						\node (a) at (5.5,0) {$b$};
						\node (b) at (5,-1) {$a$};
						\node (c) at (6,-1) {$c$};
						
						\draw[-stealth](a)to(b);
						\draw[-stealth](a)to(c);
					\end{tikzpicture}
				};
				
				\node[anchor = north, label = {[label distance =-0.6em]-90:(\textit{c}) after discussing~$a$ and~$c$}] at (0, -0.155) {
					\begin{tikzpicture}[scale=0.6, anchor = south]
						\node (a) at (0,0) {$b$};
						\node (b) at (0,-1) {$a$};
						\node (c) at (0,-2) {$c$};
						
						\draw[-stealth](a)to(b);
						\draw[-stealth](b)to(c);
						
						\node (a) at (1,0) {$b$};
						\node (b) at (1,-1) {$a$};
						\node (c) at (1,-2) {$c$};
						
						\draw[-stealth](a)to(b);
						\draw[-stealth](b)to(c);

						\node (a) at (2,0) {$b$};
						\node (b) at (2,-1) {$a$};
						\node (c) at (2,-2) {$c$};
						
						\draw[-stealth](a)to(b);
						\draw[-stealth](b)to(c);
						
						\node (a) at (3,0) {$b$};
						\node (b) at (3,-1) {$a$};
						\node (c) at (3,-2) {$c$};
						
						\draw[-stealth](a)to(b);
						\draw[-stealth](b)to(c);
						
						\node (a) at (4,0) {$b$};
						\node (b) at (4,-1) {$a$};
						\node (c) at (4,-2) {$c$};
						
						\draw[-stealth](a)to(b);
						\draw[-stealth](b)to(c);
						
					\end{tikzpicture}
				};
			\end{tikzpicture}
		\end{center}

		\noindent An obvious consensus now exists: alternative~$b$.\footnote{In our specific example, the group reached a strong consensus, fully agreeing on the complete preference order. We will later discuss much weaker notions of consensus too.} 
	\end{excon}
	
	\noindent Does MD assist with consensus beyond our example? Does the order of discussion matter? 
	In this paper, we introduce a novel framework of majority dynamics for incomplete preferences, and use analytical and experimental tools to investigate its consequences with respect to a number of established consensus notions. Our main contributions show that $(i)$~majority dynamics can damage even majority-based consensus alternatives such as the Condorcet winner (though not in some restricted preference domains); $(ii)$~a chair can control the resulting consensus by determining the order of discussion, leading to controversial collective decisions; and $(iii)$~the worst-case effects described in $(i)$ and $(ii)$ are in fact rare.
	
	Earlier work has employed notions of consensus as a way to \emph{rationalise} voting rules that minimise the distance from an ideal decision, given complete preferences \citep{elkind2015distance}. Moreover, recent research has examined incomplete preferences from various static perspectives, e.g., axiomatic ones \citep{pini2009aggregating,TerzopoulouEndrissIJCAI2019}. This paper's angle is different. Instead of applying a voting rule, we search for direct agreement within the group in the form of a consensual alternative, which leads to outcomes that are less debatable and better explainable to the participating agents (note that choosing the best voting rule is a long-standing discussion in the social choice community). The process investigated by \citet{eklund2007consensus} bears a similar motivation, but assumes that agents evaluate alternatives with fixed criteria in mind, and that a chair has the (often unjustified) power to force the change of someone's preference.  \citet{auletta2018reasoning} also study (majority-based) consensus in settings of opinion diffusion, but focus on social networks, and their results are limited to a single binary issue.
	
	Moreover, many \emph{consensus measures} have been defined to study the degree of discordance between the agents' preferences in a group, without the goal of selecting a final alternative \citep{herrera2011consensual}. These measures are usually based on distances and are associated with normative properties such as unanimity, suggesting that a group of agents with the same preferences has the highest possible consensus \citep{boschcharacterizations,alcalde2011measuring}. 
	
	Finally, manipulation by the chair has been explored under the name of \emph{agenda manipulation} in judgment aggregation \citep{DietrichJET2016}, and notions of \emph{bribery} and \emph{control} in voting \citep{faliszewski2009llull,faliszewski2015complexity}. Although related in flavour, these works are technically far from ours. 
	
	This paper proceeds as follows. Section~\ref{sec:model} defines our 
	model and the basic consensus notions. 
	Section~\ref{sec:preserving}  studies the worst-case effects of the process of majority dynamics on consensus, including special cases of restricted domains. Section~\ref{sec:controlling} solves the same exercise for the problem of consensus control by the chair, and Section~\ref{sec:experiments} includes our experimental design and results. Finally, Section~\ref{sec:conclusion} concludes. 
	

	\begin{table*}[]
		\centering
		\resizebox{\linewidth}{!}{
		\begin{tabular}{ccccccccc}
			\toprule
			\multicolumn{2}{c}{\textbf{Quantifier Over}} & \multicolumn{2}{c}{\textbf{Preserving Consensus}} & \multicolumn{2}{c}{\textbf{Losing Consensus}} & \multicolumn{2}{c}{\textbf{Generating Consensus}} & \textbf{Preserving Absence}\\
			Profiles & Orders &  Existence & Identity &  Existence &  Identity & Existence & Identity & \textbf{of Consensus}\\
			\midrule
			
			$\forall$ & $\forall$ & Def.~\ref{dfn:ex-preserving} & Def.~\ref{dfn:id-preserving}  & \xmark{} co & \xmark{} co& \xmark{} co & \xmark{} $\emptyset$ & \xmark{} $\emptyset$ \\
			
			$\forall$ & $\exists$ & Def.~\ref{dfn:ex-ccontrol}  & --- & \xmark{} co & --- & \xmark{} co & --- & \xmark{} $\emptyset$  \\
			
			$\exists$ & $\forall$ & \cmark{} co & \cmark{} co   & $\overline{\text{Def.~\ref{dfn:ex-ccontrol}}}$  & $\overline{\text{Def.~\ref{dfn:id-ccontrol}}}$ & \cmark{} $\emptyset$   & $\overline{\text{Def.~\ref{dfn:dcontrol}}}$ & \cmark{} co \\
			
			$\exists$ & $\exists$ & \cmark{} co & ---  & $\overline{\text{Def.~\ref{dfn:ex-preserving}}}$  & --- & \cmark{} $\emptyset$ & ---  & \cmark{} co \\
			\bottomrule
		\end{tabular}
		}
		\caption{Effects on consensus. Each cell points to the answer source for the relevant question. For example, according to the cell at the first row and third (fourth) column, Def.~\ref{dfn:ex-preserving} (Def.~\ref{dfn:id-preserving}) captures the idea of preserving the existence (identity) of consensus for all preference profiles and all update orders.  `---' means that the question is meaningless: we can only inquire about different identities of consensus under more than one order; `\cmark{}' and `\xmark{}' indicate a trivial positive or negative answer, respectively, by considering complete (co) or empty ($\emptyset$) preference profiles. Overlined words denote the statement of a definition in the contrapositive.}
		
		\label{tab:effects}
	\end{table*}
	
	\section{The Model} \label{sec:model}
	
	This section presents our model, terminology, and notation.
	
	\subsection{Preliminaries}\label{sec:prelim}
	
	Let $N = \{1,\dots, n\}$ be a finite set of agents, and $\alt = \{a,b,c,...\}$ a finite set of $m \geq 3$ alternatives.\footnote{Our counterexamples use the minimum number of needed alternatives, but extend to any larger number by including dummy ones.}
	Each agent~$i \in N$ reports a strict partial order $\succ_i$ (an acyclic and transitive binary relation), which we call her \emph{preference} and draw as a directed acyclic graph. We can also represent each preference as a set of pairwise rankings over alternatives. By~$ab$ we refer to the ordered pair of the alternatives~$a$ and~$b$. For instance, if agent~$i$ prefers $a$ to~$b$ and $b$ to~$c$ (and thus also prefers $a$ to~$c$ because of transitivity), her preference is the set~$\{ab,bc,ac\}$. 
	
	A special kind of incomplete preferences are \emph{strict weak orderings}: they represent alternatives ranked in different levels, with alternatives of the same level being incomparable. Concretely, $\succ_i$ is a strict weak ordering if there exists some ordered partition $S_i = (S_i^1, \dots, S_i^k)$ of $A$ such that:
	\begin{enumerate}
		\item[(i)] for any $a, b \in S_i^g$ we have $ab \not\in {\succ_i}$ and $ba \not\in {\succ_i}$, and 
		\item[(ii)] for any $S_i^g$ and $S_i^h$ such that $g < h$, it must be the case that $a \succ_i b$ for all $a \in S_i^g$ and $b \in S_i^h$. 
	\end{enumerate}
	
	\noindent  \emph{Top-truncated} and \emph{bottom-truncated} preferences  are natural types of strict weak orderings, strictly ranking a subset of the alternatives and placing the remaining---incomparable---ones below or above them, respectively \citep{baumeister2012campaigns}.\footnote{Our results for strict weak orderings hold for bottom-truncated preferences. Top-truncated preferences differ very little from complete ones in the context of consensus, and are thus less interesting.} 
	
	Then, a (possibly incomplete) preference \emph{profile}~$\prof{P} = (\succ_1, \dots, \succ_n)$ is a vector of preferences of size $n$. We denote by $\support{ab} = |\{i \in N \mid a \succ_i b\}|$ the \emph{support} of $ab$, \textit{i.e.}, the number of agents who prefer alternative $a$ over $b$ in the profile~$\prof{P}$. 
	
	There are special structural characteristics exhibited by a preference profile that  make the achievement of a collective decision an unequivocal task (recall Example~\ref{ex:incomplete-apps}).
	Each of the following notions of consensus generalises an existing one from the literature on complete preferences \citep{elkind2015distance}. We start with a traditional concept:
	
	\begin{itemize}
		\item[$\blacktriangleright$] An alternative $a$ is a \textbf{Condorcet winner} (CW)  in $\prof{P}$ if $\support{ab} > \support{ba}$ for all $b \in A\setminus \{a\}$. 
	\end{itemize}
	
	\noindent Next, we distinguish two notions that capture an alternative being on ``top''  of an incomplete preference, which are unambiguous in the complete case: An alternative~$a$ is \emph{undominated} (abbreviated `UD') in $\succ$ if there is no $b \in A$ such that $b \succ a$, and, $a$ is \emph{dominant} (abbreviated `Dom') if  $a \succ b$ for all $b \in A\setminus \{a\}$. We define consensus for the former notion, and the latter is analogous.
	
	\begin{itemize}
		\item[$\blacktriangleright$] An alternative $a$ is \textbf{unanimity undominated} (\unanundom) in $\prof{P}$ if it is undominated in $\succ_i$ for all $i \in N$.
		\item[$\blacktriangleright$] An alternative $a$ is \textbf{majority undominated} (\majundom)  in $\prof{P}$ if $|\{i \in N \mid a \text { is undominated in } \succ_i \}|> |N|/2$.
		\item[$\blacktriangleright$] An alternative $a$ is \textbf{plurality undominated} (\plurundom) in $\prof{P}$ if $|\{i \in N \mid a \text { is undominated in } \succ_i\}| >0$ and $ a \in \argmax_{b\in \alt}|\{i \in N \mid b \text { is undominated in } \succ_i\}|$.
	\end{itemize}
	
	\noindent More than one alternative in a profile~$\prof{P}$ may satisfy the  undominated consensus definitions (and the plurality  dominant one). Importantly, we only say that a profile exhibits consensus when there is a \emph{unique} such alternative. 
	
	Given a consensus notion~$C \in \{$`\CW', `\unanundom', `\unandom', `\majundom', `\majdom', `\plurundom', `\plurdom'$\}$ and a preference profile~$\prof{P}$, we define $C(\prof{P})\in \alt \cup \{\bot\}$ to be the consensus alternative in~$P$, with respect to~$C$; if such a consensus alternative does not exist, we write  $C(\prof{P})=\bot$. 
	
	
	
	\subsection{Majority Dynamics}
	
	Let $\sigma = (p_1, \dots, p_{\ell})$ be an ordering of pairs of alternatives, such that for any two alternatives $a$ and $b$, exactly one of $ab$ or $ba$ are in $\sigma$. We define the process of \emph{majority dynamics}~$\dyns{\sigma}$ such that $\dyns{\sigma}(\prof{P})$ is the stable profile 
	(denoted $\prof{P}^{|\sigma|}$ below) 
	that results from the application of the dynamics on the initial profile $\prof{P}$ where pairs are discussed following the update order $\sigma$ (first $p_1$, then $p_2$, \textit{etc}). Let $\prof{P}^0 := \prof{P}$. When the pair $p_t = ab$ is discussed, the preferences are updated from $\prof{P}^{t-1} = (\succ_i^{t - 1})_{i \in N}$ to $\prof{P}^{t} = (\succ_i^t)_{i \in N}$ such that:
	\begin{linenomath}
		\begin{equation*}
			\succ_i^{t}  = 
			\begin{cases}
				\succ^{t-1}_i  & \text{if $ab$ or $ba \in {\succ^{t-1}_i }$}
				\\
				\llbracket \succ^{t-1}_i \cup \{ab\} \rrbracket& \text{otherwise, and if } \support{ab} \geq \support{ba}\\
				\llbracket \succ^{t-1}_i \cup \{ba\} \rrbracket & \text{otherwise.}
			\end{cases}
		\end{equation*}
	\end{linenomath}
	
	\noindent Here, $\llbracket \succ \rrbracket$ is the transitive closure of $\succ$, ensuring we never violate transitivity. For example, an agent with preference set~$\{bc\}$ that adds~$ab$, will also have to add~$ac$.  This means that individual agents' hold (possibly incomplete) transitive preferences with no cycles at every step of the dynamics. The profile $ \dyns{\sigma}(\prof{P})$
	is thus a profile of complete preferences over the set of alternatives.  Note that we always break ties in favour of the first alternative in the considered pair.  
	
	\subsection{Effects on Consensus}
	
	All possible effects that $\dyn$ may have on group consensus are categorised into four types: $(i)$~preserving, $(ii)$~losing, or $(iii)$~generating consensus, and $(iv)$~preserving the absence of consensus. Moreover, given a profile of incomplete preferences, different update orders may lead to different consensus alternatives. With this in mind, effects $(i)$ to $(iii)$ further distinguish between the existence of \emph{any} consensus alternative, and the identity of a \emph{specific} consensus alternative.  
	
	Table~\ref{tab:effects} shows that five definitions below capture all non-trivial cases regarding effects on consensus, meaning that together they are expressively complete. We begin with definitions that refer to effects regarding \emph{all} possible update orders. 
	
	\begin{definition} \label{dfn:ex-preserving}
		Given a consensus notion~$C$, $\dyn$  \textbf{preserves $C$ existence}  if for all profiles~$\prof{P}$:
		\[C(\prof{P})\neq \bot \text{ implies that } C(\dyns{\sigma}(\prof{P}))\neq \bot, \text{ for all orders }\sigma.\] 

	\end{definition}
	
	\begin{definition} \label{dfn:id-preserving}
		Given a consensus notion~$C$, $\dyn$ \textbf{preserves $C$ identity} if for all profiles~$\prof{P}$ and alternatives~$a\in \alt$:
		\[C(\prof{P})= a \text{ implies that } C(\dyns{\sigma}(\prof{P}))=a, \text{ for all orders }\sigma.\] 
	\end{definition}
	
	\noindent We next weaken consensus preservation, by quantifying over \emph{some} order instead of \emph{all} orders. Positive control captures scenarios where $\dyn$ enables the mechanism designer to select a suitable update order to preserve consensus.

	\begin{definition} \label{dfn:ex-ccontrol}
		Given a consensus notion~$C$, $\dyn$ enables \textbf{positive $C$ existence control}  if  for all profiles~$\prof{P}$:
		\[C(\prof{P})\neq \bot \text{ implies } C(\dyns{\sigma}(\prof{P}))\neq \bot, \text{ for some order } \sigma.\]
	\end{definition}
	
	\begin{definition} \label{dfn:id-ccontrol}
		Given a consensus notion~$C$, $\dyn$ enables \textbf{positive $C$ identity control}  if  for all profiles~$\prof{P}$ and alt.~$a\in \alt$:
		\[C(\prof{P}) =a \text{ implies } C(\dyns{\sigma}(\prof{P}))=a, \text{ for some order } \sigma.\]
	\end{definition}
	
	\noindent While positive control enables the achievement of consensus, negative control prevents a specific consensus from forming. 
	This may be done either by imposing no consensus at the end of the dynamic process, or by inflicting different consensus alternatives, depending on the update order. 
	
	\begin{definition} \label{dfn:dcontrol}
		Given a consensus notion~C, $\dyn$ enables \textbf{negative C control}  if  for all profiles~$\prof{P}$ with $C(\prof{P})=\bot$, one of the following conditions hold:
		\begin{itemize}
			\item $C(\dyns{\sigma}(\prof{P}))=\bot$, for some order~$\sigma$;
			\item $ C(\dyns{\sigma}(\prof{P})) \neq C(\dyns{\sigma'}(\prof{P}))$, for some orders $\sigma,\sigma'$.
		\end{itemize}
	\end{definition}
	
	Table~\ref{tab:sum} summarises the results we will prove in Sections~\ref{sec:preserving} and~\ref{sec:controlling}---notably, we will see that issues regarding existence and identity coincide for our consensus notions.
	
	\begin{table}
		\centering
		\begin{tabular}{rccc}
			\toprule
			& Pres.\ consensus & Pos.\ control & Neg.\ control\\
			\toprule
			\CW & \xmark{} (\cmarkorange{}) & \cmark{}  & \cmark{} \\
			\plurundom & \xmark{}  & \xmark{} & \xmark{} \\
			\plurdom & \xmark{} & \xmark{} & \xmark{} \\
			\majundom & \xmark{} & \cmark{} & \xmark{} \\
			\majdom & \cmark{} & \cmark{} & \xmark{} \\
			\unanundom & \xmark{} (\cmarkorange{}) & \cmark{} & \cmark{} \\
			\unandom &  \cmark{} & \cmark{}  & \xmark{} \\
			\bottomrule
		\end{tabular}
		\caption{Summary of effects on consensus. `\cmark' means that the definition of the relevant effect holds for the given consensus notion, and `\xmark' means that it is violated. The results for strict weak orderings, when they differ from the general case, are shown in a parenthesis. }
		\label{tab:sum}
	\end{table}
	

	\section{Preserving Consensus} \label{sec:preserving}
	
	We examine whether $\dyn$ preserves consensus, starting with the notion of a Condorcet winner.
	For the special case of three alternatives, we can report some good news. 
	
	\begin{proposition}
		For $m = 3$, $\dyn$ preserves $\CW$~existence, but not $\CW$~identity.
	\end{proposition}
	
	\begin{proof}
		Let $A = \{a,b,c\}$, and let~$\sigma = (ab,ac,bc)$ be the update order (for different orders the proof is analogous). Consider any profile $\prof{P}^0$. If alternative~$a$ is a Condorcet winner in $\prof{P}^0$ we know that the support of $ba$ or $ca$ will not increase throughout $\dyns{\sigma}$, and $a$ will remain Condorcet winner in $\dyns{\sigma}(\prof{P}^0)$. Thus, we only need to consider alternative~$b$ or~$c$ being the Condorcet winner. As the proofs proceed similarly, we only consider one of the two cases.
		
		Suppose alternative~$b$ is the Condorcet winner in $\prof{P}^0$. When the pair~$ab$ is discussed at time~$t = 1$, the agents who update their preference will support $b$ over $a$, and transitivity requirements can lead them to support $b$ over $c$, but never $c$ over $b$.
		
		Now, if $\supportz{1}{ac} \geq \supportz{1}{ca}$, then, the support of $cb$ will not increase at time $t = 2$. We thus have~$\supportz{2}{cb} = \supportz{0}{cb} < \supportz{0}{bc} \leq \supportz{2}{bc}$. Hence, after time $t = 3$, $bc$ will have more support than $cb$ and $b$ will still be the Condorcet winner in $\dyns{\sigma}(\prof{P}^0)$. 
		
		Suppose instead that $\supportz{1}{ac} < \supportz{1}{ca}$. Then, at time $t = 2$, the support of $ca$ will increase, and potentially also that of $cb$. At time $t = ¨3$, if $\supportz{2}{bc} \geq \supportz{2}{cb}$, $b$ will remain the Condorcet winner in $\dyns{\sigma}(\prof{P}^0)$. If instead $\supportz{2}{bc} < \supportz{2}{cb}$, then $c$ will be the Condorcet winner in $\dyns{\sigma}(\prof{P}^0)$, as a majority of agents supported $c$ over $a$ at time~$t=2$. Note that this results in a change of the identity of the Condorcet winner. 
	\end{proof}
	\noindent However,  the situation is not as bright when $m>3$.
	
	\begin{proposition}
		For $m > 3$,  $\dyn$ does not preserve $\CW$~existence (thus neither $\CW$~identity).
	\end{proposition}
	\begin{proof}
		Consider the $3$-agent profile $\prof{P}$ described below (ignoring the dotted edges) where $w$ is the Condorcet winner.
		\begin{center}
			\scalebox{0.9}{
				\begin{tikzpicture}
					\node (agent1) at (0,0) {
						\begin{tikzpicture}
							\node (b) at (-0.5, 0) {$b$};
							\node (w) at (0.5, 0) {$w$};
							\node (c) at (0, -0.9) {$c$};
							\node (a) at (0, -1.8) {$a$};
							
							\path[-stealth] (b) edge (c);
							\path[-stealth] (w) edge (c);
							\path[-stealth] (c) edge (a);
							\path[-stealth,dotted] (b) edge (w);
					\end{tikzpicture}};
					
					\node (agent2) at (2,0) {
						\begin{tikzpicture}
							\node (a) at (0, 0) {$a$};
							\node (b) at (-0.5, -0.9) {$b$};
							\node (c) at (0.5, -0.9) {$c$};
							\node (w) at (0, -1.8) {$w$};
							
							\path[-stealth] (a) edge (b);
							\path[-stealth] (a) edge (c);
							\path[-stealth, dotted] (b) edge (w);
							\path[-stealth] (c) edge (w);
							\path[-stealth,dotted] (b) edge (c);
					\end{tikzpicture}};
					\node (agent3) at (4,0) {
						\begin{tikzpicture}
							\node (w) at (0, 0) {$w$};
							\node (a) at (0, -0.6) {$a$};
							\node (b) at (0, -1.2) {$b$};
							\node (c) at (0, -1.8) {$c$};
							
							\path[-stealth] (w) edge (a);
							\path[-stealth] (a) edge (b);
							\path[-stealth] (b) edge (c);
					\end{tikzpicture}};
				\end{tikzpicture}
			}
		\end{center}
		\noindent Let $\sigma$ be any update order starting with $bc$ and $bw$. $\dyns{\sigma}$ on this profile results in the inclusion of the dotted edges. Note that 
		$w$ is no longer a Condorcet winner. In fact, one can see that 
		no Condorcet winner exists in $\dyns{\sigma}(\prof{P})$. 	
	\end{proof}
	
	When it comes to undominated alternatives and $\plurdom$, we find that $\dyn$ violates consensus preservation even for $m = 3$. The opposite holds for $\unandom$ and $\majdom$.
	
	\begin{proposition}\label{prop:undomnotpreserved}
		Take $C \in\{$ `\plurundom', `\majundom', `\unanundom', `\plurdom'$\}$. Then $\dyn$ does not  preserve $C$ existence (thus neither identity).
	\end{proposition}
	
	\begin{proof}[Partial proof.]
		We start with $\plurundom$ and $\majundom$. Consider the profile below (ignoring the dotted edges) where $a$ is the plurality and majority undominated consensus. 
		
		\begin{center}
			\scalebox{0.9}{
				\begin{tikzpicture}
					\node (agent1) {
						\begin{tikzpicture}
							\node (a) at (-0.5, 0) {$a$};
							\node (b) at (0.5,0) {$b$};
							\node (c) at (0, -0.7) {$c$};
							
							\path[-stealth] (a) edge (c);
							\path[-stealth] (b) edge (c);
							\path[-stealth, dotted] (b) edge (a);
					\end{tikzpicture}};
					
					\node[right = 1em of agent1] (agent2) {
						\begin{tikzpicture}
							\node (a) at (-0.5, 0) {$a$};
							\node (c) at (0.5,0) {$c$};
							\node (b) at (0, -0.7) {$b$};
							
							\path[-stealth] (a) edge (b);
							\path[-stealth] (c) edge (b);
							\path[-stealth, dotted] (c) edge (a);
					\end{tikzpicture}};
					
					\node[right = 1em of agent2] (agent3) {
						\begin{tikzpicture}
							\node (a) at (0, -0.6) {$a$};
							\node (b) at (0, -1.2) {$b$};
							\node (c) at (0, -1.8) {$c$};

							\path[-stealth] (a) edge (b);
							\path[-stealth] (b) edge (c);
					\end{tikzpicture}};

					\node[right = 1em of agent3] (agent4) {
						\begin{tikzpicture}
							\node (b) at (0, -0.6) {$b$};
							\node (c) at (0, -1.2) {$c$};
							\node (a) at (0, -1.8) {$a$};

							\path[-stealth] (b) edge (c);
							\path[-stealth] (c) edge (a);
					\end{tikzpicture}};
					
					\node[right = 1em of agent4] (agent5) {
						\begin{tikzpicture}
							\node (b) at (0, -0.6) {$c$};
							\node (c) at (0, -1.2) {$b$};
							\node (a) at (0, -1.8) {$a$};
							
							\path[-stealth] (b) edge (c);
							\path[-stealth] (c) edge (a);
					\end{tikzpicture}};
				\end{tikzpicture}
			}
		\end{center}
		
		\noindent Let $\sigma$ be an order that starts with the sequence $(ba, ca)$. As the majority of agents prefer both $b$ and $c$ to $a$, we will end up including the dotted edges. Overall, both $b$ and $c$ will be dominant in exactly two individual rankings.
		
		Let us now turn to $\plurdom$. Consider the following profile (ignoring the dotted edges) where $a$ is the unique plurality dominant alternative.
		
		\begin{center}
			\begin{tikzpicture}	
				
				\node (agent1) at (0,0) {
					\begin{tikzpicture}
						\node (a) at (0, 0) {$a$};
						\node (b) at (0,-0.7) {$b$};
						\node (c) at (0, -1.4) {$c$};
						
						\path[-stealth] (a) edge (b);
						\path[-stealth] (b) edge (c);
				\end{tikzpicture}};
				
				\node (agent2) at (2,0) {
					\begin{tikzpicture}
						\node (a) at (0, 0) {$a$};
						\node (b) at (0,-0.7) {$b$};
						\node (c) at (0, -1.4) {$c$};
						
						\path[-stealth] (a) edge (b);
						\path[-stealth] (b) edge (c);
				\end{tikzpicture}};
				\node (agent3) at (4,0) {
					\begin{tikzpicture}
						\node (b) at (-0.4, 0) {$b$};
						\node (a) at (0,-0.6) {$a$};
						\node (c) at (0.4, 0) {$c$};
						
						\path[-stealth] (b) edge (a);
						\path[-stealth] (c) edge (a);
						\path[-stealth,dotted] (b) edge (c);
				\end{tikzpicture}};
				\node (agent4) at (6,0) {
					\begin{tikzpicture}
						\node (b) at (-0.4, 0) {$b$};
						\node (a) at (0,-0.6) {$a$};
						\node (c) at (0.4, 0) {$c$};
						
						\path[-stealth] (b) edge (a);
						\path[-stealth] (c) edge (a);
						\path[-stealth,dotted] (b) edge (c);
				\end{tikzpicture}};
				
			\end{tikzpicture}
		\end{center}
		
		Let the update order~$\sigma$ be a sequence starting with the pair $bc$. This will result in the inclusion of the dotted edges in the two last agents' preferences, and will leave us with a profile without a unique plurality dominant alternative.

		Finally, for $\unanundom$ consider the following profile (again, ignoring the dotted edges) where $a$ is the unique unanimity undominated alternative---as it is the only alternative with no incoming edges.

		\begin{center}
			\begin{tikzpicture}	
				
				\node (agent1) at (0,0) {
					\begin{tikzpicture}[anchor=south]
						\node (c) at (0, 0) {$c$};
						\node (b) at (0,-0.7) {$b$};
						\node (a) at (0.7, -0.7) {$a$};
						
						\path[-stealth] (c) edge (b);
						\path[-stealth,dotted] (b) edge (a);
				\end{tikzpicture}};
				
				\node (agent2) at (2,0) {
					\begin{tikzpicture}[anchor=south]
						\node (c) at (0, 0) {$c$};
						\node (b) at (0,-0.7) {$b$};
						\node (a) at (0.7, -0.7) {$a$};
						
						\path[-stealth] (c) edge (b);
						\path[-stealth,dotted] (b) edge (a);
				\end{tikzpicture}};
				
				\node (agent3) at (4,0) {
					\begin{tikzpicture}
						\node (a) at (0, 0) {$a$};
						\node (c) at (0,-0.7) {$c$};
						\node (b) at (0.7, 0) {$b$};
						
						\path[-stealth] (a) edge (c);
						\path[-stealth,dotted] (b) edge (a);
				\end{tikzpicture}};
				
			\end{tikzpicture}
		\end{center}
		
		For any update order $\sigma$ starting with $ab$, we include the dotted edges. So, we end up in a profile without a unanimity undominated alternative as $a$ will be the dominant alternative in the preferences of the last agent, while $c$ will be the dominant alternative in the first two agents' preferences. 
	\end{proof}

	\begin{proposition} \label{prop:dom-preserve}
		Take $C \in\{$`\unandom', `\majdom'$\}$. Then $\dyn$ preserves $C$ identity (thus also existence). 
	\end{proposition}
	
	\begin{proof}[Proof sketch.]
		Suppose $a \in \alt$ is the consensus alternative according to \unandom{} or \majdom{}. Then, any update on a pair $ab$ or $ba$, for any $b \in \alt \setminus \{a\}$, will be made in favour of $a$.
		For the same reason, no $b \in \alt\setminus \{a\}$ can be preferred to $a$ because of transitivity.
		The consensus will thus remain.
	\end{proof}
	
	
	
	\paragraph{Quality of Consensus.} In cases where consensus is preserved, but the identity of the consensus alternative is not, we explore how ``bad'' the new consensus can be. Unfortunately, the answer is not very positive. 
	
	Before diving in, we need the following definition: An alternative $a$ is a \emph{Condorcet loser}  in $\prof{P}$ if $\support{ab} < \support{ba}$ for all $b \in A\setminus \{a\}$.

	\begin{proposition}\label{prop:loserwins}
		For 
		$m \geq 5$,  $\dyn$ can turn a Condorcet loser into a Condorcet winner. 
	\end{proposition}
	
	\begin{proof}
		Let $A = \{w, \ell, a, b, x\}$, and consider the 7-agent profile~$\prof{P^0}$ below (ignoring the dotted edges), where $w$ is the Condorcet winner and $\ell$ is the Condorcet loser.
		\begin{center}
			\scalebox{1}{
				\begin{tikzpicture}
					\node (agent1) at (0,0) {
						\begin{tikzpicture}
							\node (l) at (0, 0) {$\ell$};
							\node (a) at (0, -0.7) {$a$};
							\node (x) at (0.8, -0.7) {$x$};
							\node (w) at (0.8, -1.4) {$w$};
							\node (b) at (0.8, -2.1) {$b$};
							
							\path[-stealth] (l) edge (a);
							\path[-stealth] (x) edge (w);
							\path[-stealth] (w) edge (b);
							
							\path[dotted,  -stealth] (0.2, -0.7) edge (0.6, -0.7);
					\end{tikzpicture}};
					\node (label1) at (0,-1.5) {1 agent}; 
					
					\node (agent2) at (1.7,0) {
						\begin{tikzpicture}
							\node (l) at (0, 0) {$\ell$};
							\node (b) at (0, -0.7) {$b$};
							\node (x) at (0.8, -0.7) {$x$};
							\node (w) at (0.8, -1.4) {$w$};
							\node (a) at (0.8, -2.1) {$a$};
							
							\path[-stealth] (l) edge (b);
							\path[-stealth] (x) edge (w);
							\path[-stealth] (w) edge (a);
							
							\path[dotted,  -stealth] (0.2, -0.7) edge (0.6, -0.7);
					\end{tikzpicture}};
					\node (label2) at (1.7,-1.5) {1 agent}; 	
					
					\node (agent3) at (3.4,0) {
						\begin{tikzpicture}
							\node (w) at (0, 0) {$w$};
							\node (a) at (0, -0.7) {$a$};
							\node (b) at (0, -1.4) {$b$};
							\node (x) at (0, -2.1) {$x$};
							
							\node (l) at (1, 0) {$\ell$};
							\path[-stealth] (w) edge (a);
							\path[-stealth] (a) edge (b);
							\path[-stealth] (b) edge (x);

							\path[dotted, -stealth] (l) edge (w);
					\end{tikzpicture}};
					\node (label2) at (3.4,-1.5) {2 agents}; 	
					
					\node (agent3) at (5.1,0) {
						\begin{tikzpicture}
							\node (w) at (0, 0) {$w$};
							\node (a) at (0, -0.7) {$a$};
							\node (b) at (0, -1.4) {$b$};
							\node (x) at (0, -2.1) {$x$};
							
							\node (l) at (1, 0) {$\ell$};
							\path[-stealth] (w) edge (a);
							\path[-stealth] (a) edge (b);
							\path[-stealth] (b) edge (x);

							\path[-stealth] (w) edge (l);
					\end{tikzpicture}};
					\node (label2) at (5.1,-1.5) {1 agent};

					\node (agent4) at (6.8,0) {
						\begin{tikzpicture}
							\node (a) at (0, 0) {$a$};
							\node (b) at (0, -0.7) {$b$};
							\node (x) at (0, -1.4) {$x$};
							\node (l) at (0, -2.1) {$\ell$};
							
							\node (w) at (1, -2.1) {$w$};
							\path[-stealth] (a) edge (b);
							\path[-stealth] (b) edge (x);
							\path[-stealth] (x) edge (l);

							\path[dotted, -stealth] (l) edge (w);
					\end{tikzpicture}};
					\node (label2) at (6.8,-1.5) {2 agents}; 	
					
					%
				\end{tikzpicture}
			}
		\end{center}
		\noindent Consider what happens if the update order $\sigma$ starts with $ax, bx, \ell w$. After time~$t=2$, agents have updated their preferences on $ax$ and $bx$ and $\ell$ is preferred to $w$ by a majority. Updating on $\ell w$ at time~$t=3$ makes $\ell$ the Condorcet winner, and that cannot change. The updated preferences after time $t = 3$ are represented by the dotted edges. 		
	\end{proof}
	
	
	A similar proposition holds for \plurdom{} and $\plurundom$. 

	\begin{proposition}\label{prop:plurdomloses}
		For 
		$m \geq 8$, $\dyn$ can turn a plurality dominated (or undominated) alternative into a plurality dominant one. 
	\end{proposition}

	\begin{proof}
		
		Let $A = \{w, \ell, a, b, c, x, y, z\}$, and consider the 5-agent profile~$\prof{P}^0$ below (ignoring the dotted edges), where $w$ is the plurality dominant alternative and $\ell$ is the plurality dominated alternative. In this profile, the first three agents have incomplete rankings while the last two have reported a complete strict preference over the alternatives.
		
		\begin{center}
			\begin{tikzpicture}
				\node (agent1) at (0,0) {
					\begin{tikzpicture}
						\node (l) at (0, 0) {$\ell$};
						\node (a) at (0, -0.7) {$a$};
						\node (x) at (0.8, -0.7) {$x$};
						
						\node (w) at (0.3, -1.4) {$w$};
						\node (b) at (0.8, -1.4) {$b$};
						\node (c) at (1.3, -1.4) {$c$};
						
						\node (y) at (0.8, -2.1) {$y$};
						\node (z) at (1.3, -2.1) {$z$};
						\path[-stealth] (l) edge (a);
						\path[-stealth] (x) edge (w);
						\path[-stealth] (0.8, -0.9) edge (0.8, -1.2);
						\path[-stealth] (x) edge (c);
						
						\path[dotted,  -stealth] (0.2, -0.7) edge (0.6, -0.7);
						\path[dotted,  -stealth] (0.8, -1.6) edge (0.8, -1.9);
						\path[dotted,  -stealth] (c) edge (z);
				\end{tikzpicture}};
				
				\node (agent2) at (2,0) {
					\begin{tikzpicture}
						\node (l) at (0, 0) {$\ell$};
						\node (b) at (0, -0.7) {$b$};
						\node (y) at (0.8, -0.7) {$y$};
						
						\node (w) at (0.3, -1.4) {$w$};
						\node (a) at (0.8, -1.4) {$a$};
						\node (c) at (1.3, -1.4) {$c$};
						
						\node (x) at (0.8, -2.1) {$x$};
						\node (z) at (1.3, -2.1) {$z$};
						\path[-stealth] (l) edge (b);
						\path[-stealth] (y) edge (w);
						\path[-stealth] (y) edge (a);
						\path[-stealth] (y) edge (c);
						
						\path[dotted, -stealth] (a) edge (x);
						\path[dotted,  -stealth] (b) edge (y);
						\path[dotted, -stealth] (c) edge (z);
				\end{tikzpicture}};
				
				\node (agent3) at (4,0) {
					\begin{tikzpicture}
						\node (l) at (0, 0) {$\ell$};
						\node (c) at (0, -0.7) {$c$};
						\node (z) at (0.8, -0.7) {$z$};
						
						\node (w) at (0.3, -1.4) {$w$};
						\node (a) at (0.8, -1.4) {$a$};
						\node (b) at (1.3, -1.4) {$b$};
						
						\node (x) at (0.8, -2.1) {$x$};
						\node (y) at (1.3, -2.1) {$y$};
						\path[-stealth] (l) edge (c);
						\path[-stealth] (z) edge (w);
						\path[-stealth] (z) edge (a);
						\path[-stealth] (0.92, -0.9) edge (1.2, -1.2);
						
						\path[dotted, -stealth] (a) edge (x);
						\path[dotted,  -stealth] (1.3, -1.6) edge (1.3, -1.9);
						\path[dotted, -stealth] (c) edge (z);
				\end{tikzpicture}};
				
				\node (agent4) at (6,0) {
					\begin{tikzpicture}
						\node (w) at (0, 0) {$w$};
						\node (a) at (0,-0.6) {$a$};
						\node (x) at (0, -1.2) {$x$};
						\node (b) at (0,-1.8) {$b$};
						\node (y) at (0, -2.4) {$y$};
						\node (c) at (0, -3) {$c$};
						\node (z) at (0, -3.6) {$z$};
						\node (l) at (0, -4.2) {$\ell$};
						
						\path[-stealth] (w) edge (a);
						\path[-stealth] (a) edge (x);
						\path[-stealth] (0, -1.4) edge (0,-1.6) ;
						\path[-stealth] (0, -2) edge (0,-2.2) ;
						\path[-stealth] (0, -2.6) edge (0,-2.8) ;
						\path[-stealth] (c) edge (z);
						\path[-stealth] (0, -3.8) edge (0,-4) ;
				\end{tikzpicture}};
				
				\node (agent5) at (7,0) {
					\begin{tikzpicture}
						\node (w) at (0, 0) {$w$};
						\node (a) at (0,-0.6) {$a$};
						\node (x) at (0, -1.2) {$x$};
						\node (b) at (0,-1.8) {$b$};
						\node (y) at (0, -2.4) {$y$};
						\node (c) at (0, -3) {$c$};
						\node (z) at (0, -3.6) {$z$};
						\node (l) at (0, -4.2) {$\ell$};
						
						\path[-stealth] (w) edge (a);
						\path[-stealth] (a) edge (x);
						\path[-stealth] (0, -1.4) edge (0,-1.6) ;
						\path[-stealth] (0, -2) edge (0,-2.2) ;
						\path[-stealth] (0, -2.6) edge (0,-2.8) ;
						\path[-stealth] (c) edge (z);
						\path[-stealth] (0, -3.8) edge (0,-4) ;
				\end{tikzpicture}};
			\end{tikzpicture}
		\end{center}
		\noindent Consider what happens if the update order~$\sigma$ starts with the sequence $(ax,by,cz)$. After the agents have updated their preferences on these pairs (represented with dotted edges), we are left with a profile~$\prof{P}^3$ where $\ell$ is the plurality dominant alternative. No matter how the remaining preferences are updated, $\ell$ will remain the plurality dominant alternative. 
	\end{proof}
	
	
	A unanimity (majority) \emph{undominated} alternative can become a dominated alternative in all agents' preferences as this is a relatively weak consensus requirement. 
	We will thus focus on \unandom{} and \majdom{}. We say an alternative is \emph{unanimity (majority) dominated} if it is dominated by all alternatives in the preferences of all (a majority of) of agents. 
	
	\begin{proposition}
		A unan.\ (majority) dominated alternative cannot become the unan.\ (majority) dominant alternative. 
	\end{proposition}
	
	\begin{proof}[Proof sketch.]
		A majority dominated alternative is dominated by all alternatives in the preferences of at least half the agents, hence no update can be made in its favour, precluding it from becoming dominant in a majority 
		of the 
		rankings.
	\end{proof}

	\paragraph{Restricted Preference Domains.}	When reducing our scope to strict weak orderings, we find good news.
	
	\begin{proposition}\label{prop:weakordCW}
		For profiles of strict weak orderings, $\dyn$ is $\CW$-identity preserving. 
	\end{proposition}
	
	\begin{proof}
		Let $w$ be the Condorcet winner in $\prof{P}^0$. Consider a given agent~$i$, and some alternative $a \neq w$. We claim that it is not possible to have $a \succ^t_i w$ at any time step~$t$ of the majority dynamics, unless $a \succ^0_i w$. The proof is by induction.

Suppose our claim holds for all rounds $t'<t$ and all alternatives~$z\neq w$. Then $w$ is a Condorcet winner in round~$t-1$.  If agent~$i$ has already have expressed either that $a \succ_i w$ or that $w \succ_i a$ in a round before~$t$, no update can be done on this pair, so we need only consider the case where $a,w \in S_i^g$ for some $g$, where $S_i^g$ is the $g$-th tier of $i$'s preferences. There are two cases we need to examine for the update in round~$t$. 
		
		Suppose the update is performed on $aw$ or $wa$ directly. Then, $i$ will support $wa$ since $w$ is a Condorcet winner. 
		
		Suppose instead we are updating on a pair that is neither $aw$ nor $wa$, leading $i$ to prefer $a$ over $w$ because of transitivity. 
	
		Then, there must exist some alternative $b$ such that $a \succ_i^{t-1} b$ and $b \succ_i^{t-1} w$. This, however, contradicts our assumption that $\succ_i$ is a strict weak ordering, as this---in combination with our assumption that $aw \not\in {\succ_i^{t-1}}$ and our induction hypothesis---precludes the existence of an alternative that is dominated by one of $a$ and $w$, while dominating the other. 
		We conclude that if a Condorcet winner exists in $\prof{P}^0$, then that alternative must remain a Condorcet winner in $\dyns{\sigma}(\prof{P}^0)$.
	\end{proof}
	
	\noindent The proof of Proposition~\ref{prop:undomnotpreserved} uses strict weak orderings, except for $\unanundom$, and that of Proposition~\ref{prop:dom-preserve} also holds with strict weak orderings. Only the case of $\unanundom$ is left then.
	\begin{proposition}
		\label{prop:unanUDTiered}
		For profiles of strict weak orderings, $\dyn$ is $\unanundom$-identity preserving.
	\end{proposition}
	\begin{proof}
		Consider a profile $\prof{P}^0$ that represents strict weak orderings. Let $\alt' = \bigcup_{i \in N} S_i^1$ be the set of all alternatives that are in the top tier of at least one agent. 
		
		Suppose there exists a unique unanimity undominated alternative $w$ in $\prof{P}$. Then for any $a \in \alt' \setminus \{w\}$, there exists at least one $i \in N$ such that $w \succ_i a$, and that for no agent $i' \in N$ we  have $a \succ_{i'} w$. Hence, for all update order $\sigma$, when updating on the pair $wa$ or $aw$, agents will update in favour of $w$ and that for all $a \in \alt' \setminus \{w\}$. Moreover, since preferences are strict weak orderings, all alternatives in $\alt \setminus \alt'$ are initially unanimously dominated by $w$. Overall, for any update order $\sigma$, $w$ will be the unanimity undominated consensus in $\dyns{\sigma}(\prof{P}^0)$.
	\end{proof} 
	%

	
	
	\section{Controlling Consensus} \label{sec:controlling}
	
	We now investigate control issues. Note that all our proofs in this section use profiles consisting only of strict weak orderings. Thus, our results hold also in this restricted case. 
	
	We start by exploring positive control. 
	
	\begin{proposition}
		Take $C \in \{$`\CW', `\unanundom', `\unandom', `\majdom$'\}$. Then $\dyn$ enables positive $C$ identity (and thus also existence) control.
	\end{proposition}
	
	\begin{proof}[Proof sketch]
		For $\unandom$ and $\majdom$, the statement is immediate as $\dyn$ preserves the identity of consensus.
		For $\CW$ and $\unanundom$, call $w$ the consensus alternative. We claim that for the update order $\sigma$ that starts with all the pairs $wa$, for every $a \in \alt \setminus \{w\}$, the consensus is preserved by $\dyns{\sigma}$.
	\end{proof}
	
	\begin{proposition}\label{prop:plur-control}
		Take $C \in \{$`\plurundom', `\plurdom', `\majundom'$\}$. Then $\dyn$ does not enable positive $C$ existence (and thus neither identity) control.
	\end{proposition}

	\begin{proof}
		For $\plurdom$, consider the profile below (this is the same profile that appears in the proof of Proposition~\ref{prop:undomnotpreserved}).
		\begin{center}
			\begin{tikzpicture}

				\node (agent1) at (0,0) {
					\begin{tikzpicture}
						\node (a) at (0, 0) {$a$};
						\node (b) at (0,-0.7) {$b$};
						\node (c) at (0, -1.4) {$c$};
						
						\path[-stealth] (a) edge (b);
						\path[-stealth] (b) edge (c);
				\end{tikzpicture}};
				
				\node (agent2) at (2,0) {
					\begin{tikzpicture}
						\node (a) at (0, 0) {$a$};
						\node (b) at (0,-0.7) {$b$};
						\node (c) at (0, -1.4) {$c$};
						
						\path[-stealth] (a) edge (b);
						\path[-stealth] (b) edge (c);
				\end{tikzpicture}};
				\node (agent3) at (4,0) {
					\begin{tikzpicture}
						\node (b) at (-0.4, 0) {$b$};
						\node (a) at (0,-0.6) {$a$};
						\node (c) at (0.4, 0) {$c$};
						
						\path[-stealth] (b) edge (a);
						\path[-stealth] (c) edge (a);
						\path[-stealth,dotted] (b) edge (c);
				\end{tikzpicture}};
				\node (agent4) at (6,0) {
					\begin{tikzpicture}
						\node (b) at (-0.4, 0) {$b$};
						\node (a) at (0,-0.6) {$a$};
						\node (c) at (0.4, 0) {$c$};
						
						\path[-stealth] (b) edge (a);
						\path[-stealth] (c) edge (a);
						\path[-stealth,dotted] (b) edge (c);
				\end{tikzpicture}};
				
			\end{tikzpicture}
		\end{center}
		
		\noindent In this profile alternative $a$ is the unique plurality dominant alternative. Note that only one update is possible in this profile, independent of the update order $\sigma$. In the resulting profile $b$ will be the dominant alternative in two agents' preferences, no matter the update order, meaning there will no longer be a unique plurality dominant alternative.

		For $\plurundom$ and $\majundom$ consider the following profile (ignoring dotted edges), where $a$ is the unique plurality and majority undominated alternative.
		%
		%
		%
		%
		%
		%
		%
		%
		%
		%

		\begin{center}
			\begin{tikzpicture}
				\node (agent1) at (0,0) {
					\begin{tikzpicture}
						\node (c) at (0, 0) {$c$};
						\node (b) at (0,-0.7) {$b$};
						\node (a) at (0, -1.4) {$a$};
						
						\path[-stealth] (c) edge (b);
						\path[-stealth] (b) edge (a);
				\end{tikzpicture}};
				\node (label1) at (0,-1.2) {2 agents}; 
				
				\node (agent2) at (2,0) {
					\begin{tikzpicture}
						\node (a) at (0, 0) {$a$};
						\node (c) at (0,-0.7) {$c$};
						\node (b) at (0, -1.4) {$b$};
						
						\path[-stealth] (a) edge (c);
						\path[-stealth] (c) edge (b);
				\end{tikzpicture}};
				\node (label2) at (2,-1.2) {1 agents}; 
				
				\node (agent3) at (4,0) {
					\begin{tikzpicture}
						\node (a) at (-0.4, 0) {$a$};
						\node (c) at (0,-0.6) {$c$};
						\node (b) at (0.4, 0) {$b$};
						
						\path[-stealth] (b) edge (c);
						\path[-stealth] (a) edge (c);
						\path[-stealth, dotted] (b) edge (a);
				\end{tikzpicture}};
				\node (label3) at (4,-1.2) {2 agents}; 
			\end{tikzpicture}
		\end{center}
		Here, there is again only one possible update, resulting to $b$ and $c$ being undominated in the preferences of exactly two agents, and one agent for $a$. Overall no alternative is majority undominated, nor uniquely plurality undominated.
	\end{proof}
	%

	Negative control is more difficult to achieve.
	
	\begin{proposition}
		Take $C \in \{$`\CW', `\unanundom'$\}$. Then $\dyn$ enables negative C control.
	\end{proposition}
	
	\begin{proof}[Proof]
		We first consider the case of $\CW$. Let $\prof{P}$ be a profile without a Condorcet winner. Then, for 
		any alternative $a \in \alt$, there is another alternative $b \in \alt \setminus \{a\}$ such that $\support{ba} \geq \support{ab}$. Now, for the update order $\sigma$ starting with the pair $ba$, we must have $\CW(\dyns{\sigma}(\prof{P})) \neq a$. If there is no Condorcet winner in $\dyns{\sigma}(\prof{P})$, we are done. In the other case, let $w = \CW(\dyns{\sigma}(\prof{P}))$. Because the previous also applies to $w$, there is $\sigma'$ such that $w$ is not the Condorcet winner in $\dyns{\sigma'}(\prof{P})$. Now if $\CW(\dyns{\sigma'}(\prof{P})) = \bot$, we are done. Otherwise, we have two update orders $\sigma$ and $\sigma'$ such that $\CW(\dyns{\sigma}(\prof{P})) \neq \CW(\dyns{\sigma'}(\prof{P}))$, which proves the statement.
		
		For the case of \unanundom{}, it is clear that for any profile $\prof{P}$ with no unanimity undominated alternative, no such alternative can arise at the end of $\dyn$. If, on the other hand, there are two unanimity undominated alternatives $a$ and $b$ in $\prof{P}$, then for $\sigma$ starting with $ab$, we will have $\unanundom(\dyns{\sigma}(\prof{P})) = a$. However, for $\sigma'$ starting with $ba$, we would obtain $\unanundom(\dyns{\sigma'}(\prof{P})) = b \neq \unanundom(\dyns{\sigma}(\prof{P}))$.
	\end{proof}
	
	\begin{proposition}
		\label{prop:noNegativeControl}
		Take $C \in \{$ `\unandom', `\plurundom', `\plurdom', `\majundom', `\majdom'$\}$. Then $\dyn$ does not enable negative $C$ control.
	\end{proposition}

	\begin{proof}
		
		For $\plurdom$, $\majdom$, $\majundom$ and $\unandom$, consider the following profile (ignoring the dotted edges) where no unique consensus alternative exists for any of the three notions. 
		
		\begin{center}
			\begin{tikzpicture}
				\node (agent1) at (0,0) {
					\begin{tikzpicture}
						\node (b) at (-0.4, 0) {$b$};
						\node (a) at (0,-0.6) {$a$};
						\node (c) at (0.4, 0) {$c$};
						
						\path[-stealth] (b) edge (a);
						\path[-stealth] (c) edge (a);
						\path[-stealth,dotted] (b) edge (c);
					\end{tikzpicture}
				};
				\node (agent2) at (2,0) {
					\begin{tikzpicture}
						\node (a) at (0.4, 0) {$a$};
						\node (c) at (0,-0.6) {$c$};
						\node (b) at (-0.4, 0) {$b$};
						
						\path[-stealth] (a) edge (c);
						\path[-stealth] (b) edge (c);
						\path[-stealth,dotted] (b) edge (a);
					\end{tikzpicture}
				};
				
				\node (agent3) at (4,0) {
					\begin{tikzpicture}
						\node (a) at (0.4, 0) {$a$};
						\node (c) at (0,-0.6) {$c$};
						\node (b) at (-0.4, 0) {$b$};
						
						\path[-stealth] (a) edge (c);
						\path[-stealth] (b) edge (c);
						\path[-stealth,dotted] (b) edge (a);
					\end{tikzpicture}
				};
			\end{tikzpicture}
		\end{center}
		For any update order $\sigma$, we will end up in a profile where $b$ is the unanimity dominant alternative (and thus also the majority and plurality dominant alternative). 
		
		For $\plurundom$ consider the following profile (ignoring the dotted edges) where no unique plurality undominated alternative exists. 
		
		\begin{center}
			\begin{tikzpicture}

				\node (agent1) at (0,0) {
					\begin{tikzpicture}
						\node (c) at (0, 0) {$c$};
						\node (a) at (0,-0.6) {$a$};
						\node (b) at (0, -1.2) {$b$};
						
						\path[-stealth] (c) edge (a);
						\path[-stealth] (a) edge (b);
				\end{tikzpicture}};
				
				\node (agent1b) at (1,0) {
					\begin{tikzpicture}
						\node (b) at (0, 0) {$b$};
						\node (a) at (0,-0.6) {$a$};
						\node (c) at (0, -1.2) {$c$};
						
						\path[-stealth] (b) edge (a);
						\path[-stealth] (a) edge (c);
				\end{tikzpicture}};

				\node (agent3) at (4,0) {
					\begin{tikzpicture}
						\node (a) at (-0.4, 0) {$a$};
						\node (b) at (0,-0.6) {$b$};
						\node (c) at (0.4, 0) {$c$};
						
						\path[-stealth] (a) edge (b);
						\path[-stealth] (c) edge (b);
						\path[-stealth,dotted] (a) edge (c);
				\end{tikzpicture}};
				
				\node (agent2) at (2,0) {
					\begin{tikzpicture}
						\node (a) at (0, 0) {$a$};
						\node (b) at (0,-0.6) {$b$};
						\node (c) at (0, -1.2) {$c$};
						
						\path[-stealth] (a) edge (b);
						\path[-stealth] (b) edge (c);
				\end{tikzpicture}};
				
			\end{tikzpicture}
		\end{center}
		It is easy to see that any update will result in $a$ becoming the unique plurality undominated alternative. 
	\end{proof}
	
	
	\section{Experimental Analysis} \label{sec:experiments}
	
	The previous two sections showed how diverse the outcomes of $\dyn$ can be. We complement this formal analysis by an experimental one in order to quantify the different effects.\footnote{The experiments have been coded in Python, run on a Debian machine with 16 cores and 16GB RAM. The code and the data is available at \href{https://github.com/Simon-Rey/Let-s-Agree-to-Agree-Majority-Dynamics-for-Incomplete-Preferences}{https://github.com/Simon-Rey/Let-s-Agree-to-Agree-Majority-Dynamics-for-Incomplete-Preferences}.}
	
	For each experiment, we display the results for all consensus notion. There are overall very few distinctions between dominant and undominated-based consensus notions, so we only present the undominated case. Moreover, unanimity-based consensus imposes such a strong requirement that the relevant plots are uninformative, and have thus been omitted.
	
	\subsection{Quantifying the Effects on Consensus}

	\begin{figure*}
		\centering
		\includegraphics[width = \linewidth]{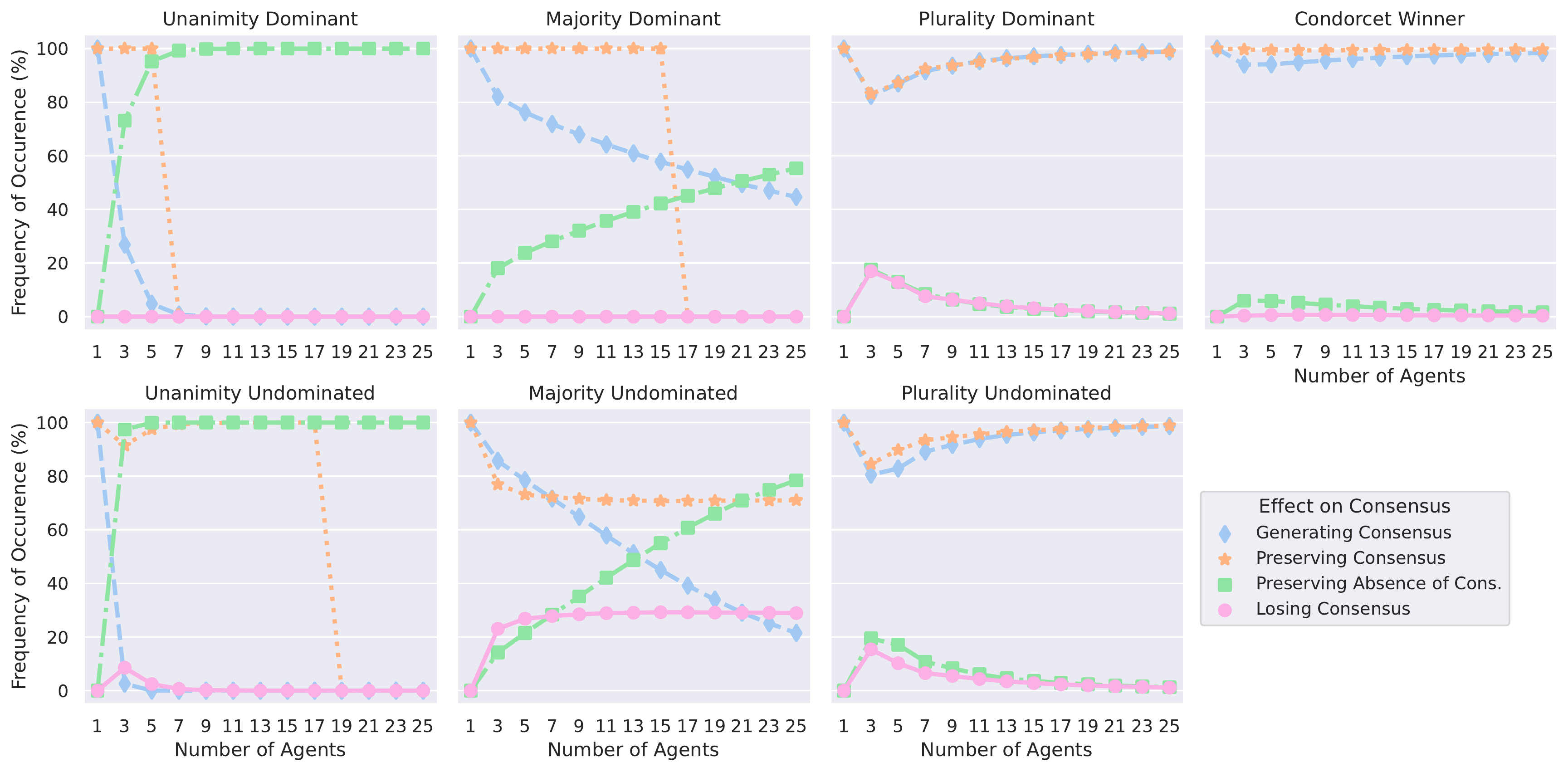}
		\caption{Frequency of each effect on consensus with respect to the number of agents.
			For preservation and loss of consensus, we normalized over the number of profiles with initial consensus, while for the other two effects over the number of profiles without initial consensus.}
		\label{fig:frequencyExpe}
	\end{figure*}

	\begin{figure*}
		\includegraphics[width=\linewidth]{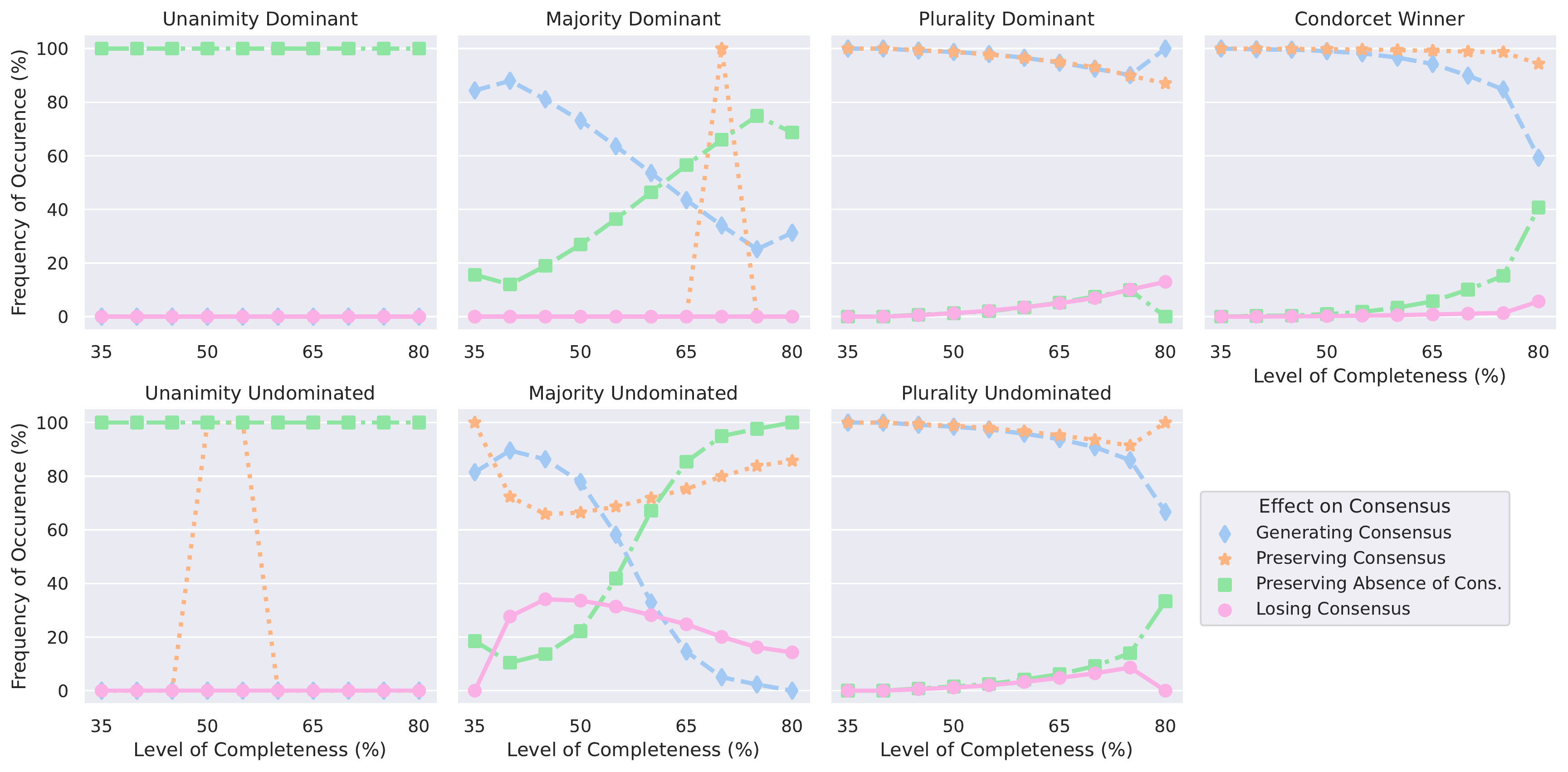}
		\caption{Frequency at which $\dyn$ has each effect on consensus with respect to the completeness level of the profiles. Only profiles with 15 agents were considered, and completeness levels have been rounded to the closest multiple of 5.
			For preservation and loss of consensus, we normalized over the number of profiles with initial consensus; while for the other two effects over the number of profiles without initial consensus.}
		\label{fig:frequencyExpeCompletenessFull}
	\end{figure*}
	
	We first explore the effects that $\dyn$ has on consensus, including generation, preservation, and loss of consensus, as well as preservation of the absence of consensus. We applied $\dyn$ on synthetic profiles and used a fixed update order.
	
	We varied the number of agents from 1 to 25, only considering odd numbers to avoid effects that rely solely on ties. For each case, we generated 5\,000\,000 random profiles over five alternatives. To generate a partial preference, we considered each pair of alternatives $\{a, b\}$ and decided with uniform probability whether $a$ beats $b$, $b$ beats $a$, or, $a$ and $b$ are not compared. We repeated this process until the generated set of pairwise comparisons was transitive. For a profile with $n$ agents, $n$ i.i.d.\ partial preferences were generated.
	
	Figure~\ref{fig:frequencyExpe} presents the complete effects that $\dyn$ can have on consensus depending on the number of agents. The first observation to make about this figure is that $\dyn$ performs particularly well according to both Condorcet  and plurality undominated consensus: consensus is almost always preserved if it was there to begin with, and generated otherwise. This trend is robust with respect to the number of agents. The picture is slightly different for majority undominated consensus. Even though $\dyn$ preserves consensus with high frequency, the frequency of generating consensus drops significantly with the number of agents. This is likely due to the fact that the majority threshold increases with the number of agents, making the consensus requirement more demanding. On the other hand, Condorcet and plurality-based requirements only depend on the profile and not on external parameters, hence their robustness. The jumps in the curves for ``Preserving Consensus'' for \unandom{}, \unanundom{} and \majdom{} are due to the very small number of random profiles initially presenting consensus.
	
	We also explored the impact that the level of completion has on each effect for all profiles with 15 agents. We computed the level of completeness as the proportion of pairwise comparisons actually provided in the profile. Figure \ref{fig:frequencyExpeCompletenessFull} presents the results. Perhaps surprisingly, the level of completeness does not significantly impact the effect on Condorcet and plurality-based consensus. Majority undominated consensus follows a more expected path: the higher the completion level, the less preferences are susceptible to change and thus the harder it is to generate consensus.
	
	
	\subsection{Quantifying the Opportunities for Control}
	
	\begin{figure*}
		\centering
		\includegraphics[width = \linewidth]{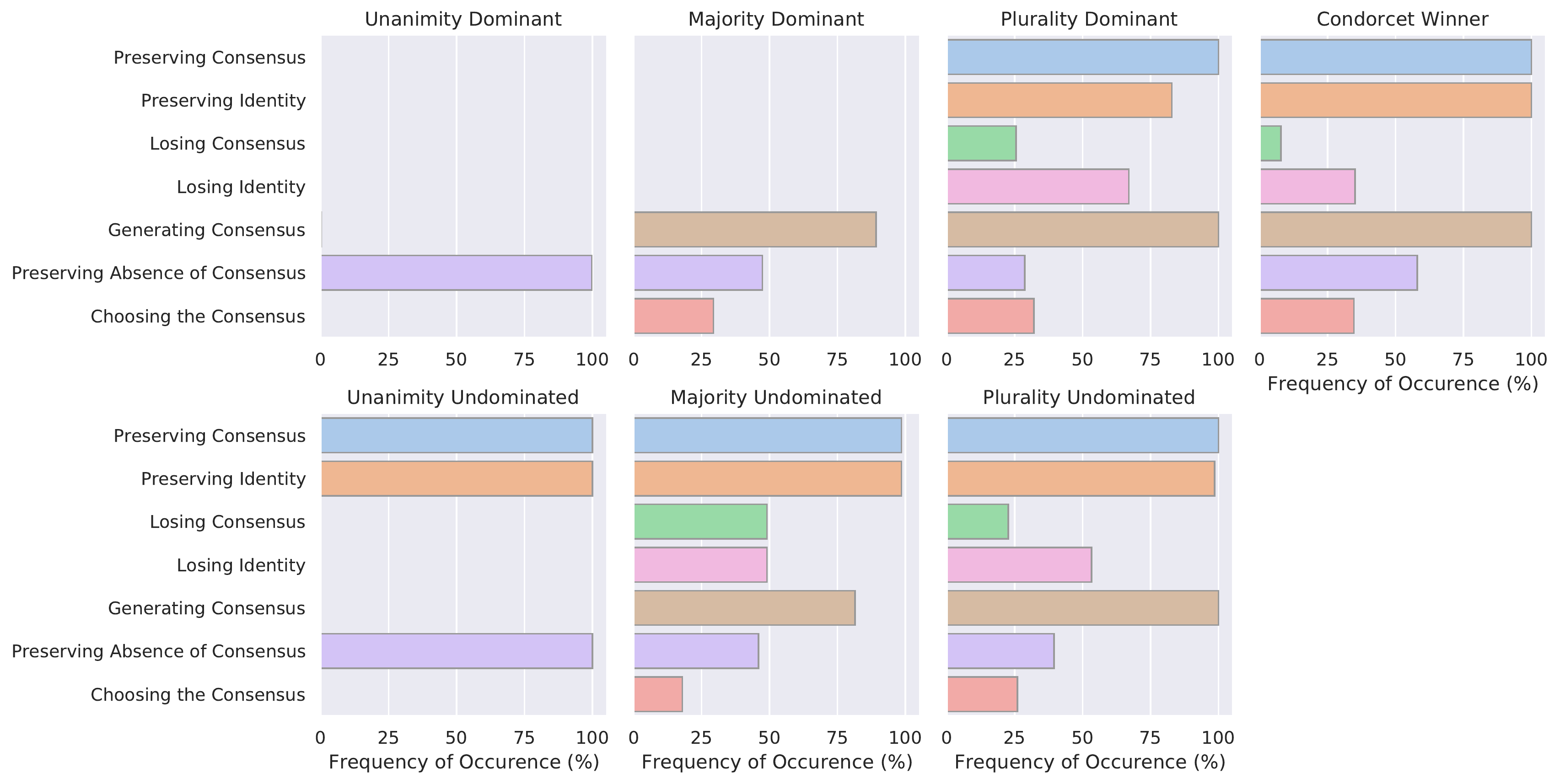}
		\caption{Frequency of each possible type of control. For the first four types, we normalized over the number of profiles with initial consensus; for the next two over the number of profiles without initial consensus; and for the last one over the total number of profiles. ``Choosing the consensus'' refers to the idea of making a specific alternative, say $a$, the consensus.}
		\label{fig:manipulation}
	\end{figure*}
	
	Through experiments, we also explore
	the question of controlling the consensus. Our goal is to quantify the power that the chair is granted by choosing the update order.
	
	We generated 20\,000 profiles of 11 agents over 4 alternatives. We run $\dyn$ on each profile and for all of the 46\,080 update orders. We counted for how many profiles at least one update order would allow the chair to achieve each type of control. The results are presented in Figure \ref{fig:manipulation}.
	
	We observed that $\dyn$ performs relatively well at preventing the chair from damaging the consensus: Only on few occasions can the chair make the consensus disappear or prevent it from being generated. A benevolent chair will however be much more successful, as it is very often possible to preserve the consensus and to generate it.
	
	Note here that for $\majdom$ it seems that for no profile can the chair select an update order to either preserve or lose consensus. This is actually not true but is a side effect of the number of profiles we generated (only 20\,000 as we had to run them on over 46\,000 update orders). Indeed, none of the sampled profiles had a majority dominant consensus at the beginning, leading to no observation of preserving or losing that consensus.
	
	
	\section{Conclusion}  \label{sec:conclusion}
	We have studied an original process of majority dynamics for agents with incomplete preferences. We have asked whether consulting the majority to fill missing opinions assists group consensus, and have answered that in the worst case it does not---only alternatives that are dominant for at least half of the agents are safe.  Countering this, we have provided some good news: Majority dynamics always preserve a Condorcet winner within natural profiles of 
	strict weak orderings, and also do so frequently within arbitrary ones. In addition, the chair of the process always has the power to choose an order of discussion such that consensus is preserved (unless we care about plurality-based consensus), and she can  very often generate a new one too (while she can rarely make an existing one disappear). Finally, if the chair plans to make a specific alternative the consensus, our experiments indicate that she can rarely choose an order to achieve it. 
	
	Yet, many questions remain open. For instance: What is the computational complexity of selecting a suitable order for control, or of minimising the number of updates until consensus is reached? In cases where consensus is not achieved, how far is the resulting profile from being consensual? These and other questions are left for future work.

\newpage

\section*{Acknowledgements}
Zoi Terzopoulou is a fellow of  the Paris Region Fellowship Programme, supported by the Paris Region. This project has received funding from the European Union's Horizon 2020 research and innovation programme under the Marie Sklodowska-Curie grant agreement No 945298-ParisRegionFP. Simon Rey is supported by NWO Vici grant 639.023.811 (``Collective Information’’). Sirin Botan is supported by the ARC Laureate Project FL200100204 ("Trustworthy AI")
	
	\bibliographystyle{named}
	\bibliography{consensus_incomplete_bib}
	
\end{document}